\documentclass[12pt]{article}
\usepackage[utf8]{inputenc}
\usepackage{evan}

\title{A Pseudopolynomial Algorithm to Minimize Maximum Lateness on Multiple Related Machines}
\author{Elbert Du\thanks{Harvard University} \and Stan Zhang\thanks{Massachussetts Institute of Technology}}
\date{\today}

\begin{document}

\maketitle

\begin{abstract}

In this paper, we will find a pseudopolynomial algorithm to solve $Qm \mid \mid L_{\max}$ and then we will prove that it is impossible to get any constant-factor approximation in polynomial time, and thus also impossible to have a PTAS for this problem. We will also show that the the problem when we don't assume a fixed number of machines, $P \mid \mid L_{\max}$, is strongly NP-hard.
\end{abstract}

\section{Introduction}

With job scheduling, since there are many problems, we have some shorthand notation for denoting problems. When we write $A \mid B \mid C$, $A$ refers to the model we are using and $B$ refers to any constraints we have, and $C$ refers to the objective we wish to minimize/maximize. We will now introduce the specific models, constraints, and objectives that will be used in this paper.\\

The objective function we will be optimizing is $L_{\max}$, which is to minimize the maximum lateness. The other objective function that will be mentioned is $C_{\max}$, which is to minimize the maximum completion time. The models we will be using are $P$ and $Q$. $P$ stands for parallel machines, all of which run at the same speed and in parallel. $Q$ refers to related machines, which means the machines still run in parallel but they can run at different speeds. However, the ratio between the time it takes to run two jobs $a$ and $b$ is constant among all machines. The number after it refers to the number of machines used, and we assume it to be a constant when we write $m$ machines. We will not be adding any constraints.\\

It is known that $P2 \mid \mid C_{\max}$ is NP-hard.\cite{Complexity} This is a generalized version of $P2 \mid \mid L_{\max}$, so this problem is NP-hard as well. Furthermore, if we add some constraints to $P2 \mid \mid L_{\max}$, it becomes strongly NP-hard.\cite{Complexity2} \cite{Scheduling}

\section{Results}
In this paper we will provide a pseudopolynomial algorithm solving the job scheduling problem of $Qm \mid \mid L_{\max}$. This is a generalization of $P2 \mid \mid L_{\max}$, and thus fills in the gap between the above results by showing that $P2 \mid \mid L_{\max}$ is in $P_{pseudo}$\\

Firstly, we will reduce solving $\min L_{\max}$ to checking feasibility on a set of jobs.

\begin{theorem}
Given an algorithm that solves checking feasibility of solving a set of jobs on time in time $T(n)$, and that $T$ is the maximum time it takes to finish any job, we can solve $\min L_{\max}$ in time $O(\log nT) T(n)$
\end{theorem}
\begin{proof}
Note that if we have an algorithm to check whether the set of jobs is feasible which runs in time $T(n)$, we can check if getting a maximum lateness of $x$ is possible by just adding $x$ to all of the deadlines and check whether the resulting set of jobs is feasible.\\

By doing this, we can now binary search for the maximum lateness. If we assume the times are all integers, we can check binary search between a lower bound of $0$ and an upper bound of $nT$ where $n$ is the number of jobs and $T$ is the maximum time it takes to run any one job. This gives us a runtime of $O(\log nT) \cdot T(n)$ to get $\min L_{\max}$.\\
\end{proof}

\subsection{Solving $P2 \mid L_{\max}$}
Firstly, we note that on any given machine, any feasible ordering of the jobs done on that machine can also be done feasibly by doing all the jobs in order of increasing deadlines (as if $2$ consecutive jobs are out of order, we can swap them).\cite{jackson} Also note that we can just scale the deadlines and the time it takes to do each job by the same amount and not change the problem, so we can assume the numbers to be integers.

Now, given two machines, we can compute feasibility as follows:\\

We DP on time and number of jobs seen: for a time $t$ and being given the first $i$ jobs, we want to have $D_{t,i} = 1$ iff there is a feasible scheduling that makes one of the machines finish at exactly time $t$ while processing all of the first $i$ jobs on time and $D_{t,i} = 0$ otherwise.\\

To determine the value of $D_{t,i}$, there must either be a scheduling where we put the $i^{th}$ job at the end of the schedule of the machine that ends up completing at time $t$ or in the other machine. To do this, we check $D_{t-t_i, i-1}$ and $D_{t,i-1}$. If $D_{t-t_i, i-1} = 1$ and $t \le d_i$, then we can add the $i^{th}$ job to the end of that machine and get a feasible scheduling of the first $i$ machines. If not, adding $i$ to the end of that machine does not give us a feasible scheduling. Otherwise, if $D_{t,i-1} = 1$, then we know that the completion time of the other machine on the first $i-1$ jobs is
$$\left(\sum_{j=1}^{i-1} t_j\right) - t$$
and that the other machine completes all of these jobs feasibly. If this value is at most $d_i - t_i$, then we can add the $i^{th}$ job to the end of it to get a valid scheduling.\\

Thus, we let $D_{t,i} = 1$ if either of the above conditions are true and $0$ otherwise. This is correct since there are only two machines, and one of them has to complete at time $t$ for this to be $1$. The above checks both possibilities for whether they are possible, and so if neither works, then we know this is impossible.\\

Now, to check feasibility, we just need to look at the $n^{th}$ row and see if any of the values are $1$. If any are, we return $1$. Otherwise, we return $0$.\\

The runtime of this algorithm is as follows:\\

We know that the columns of the DP table only need to go up to $\sum_{i=1}^n t_n \le nT$ where $T = \max t_n$ so we can make our table $n \times nT$.\\

Now, to fill out the $i^{th}$ row, we first compute $\left(\sum_{j=1}^{i-1} t_j\right)$ in $O(n)$ time. There are $n$ rows so all of these together takes $O(n^2)$ time. Then, to fill in the value at each cell in that row, we check two values in the row above and then sometimes compute whether $\left(\sum_{j=1}^{i-1} t_j\right) - t \le d_i - t_i$. Since we know the value of $\left(\sum_{j=1}^{i-1} t_j\right)$, each of these takes $O(1)$ time so we spend $O(1)$ time on each cell. There are $n^2T$ cells so this takes $O(n^2T)$ time.\\

Thus, the total runtime to check feasibility is $O(n^2T)$ and the runtime to solve $P2 \mid L_{\max}$ is $O(n^2 T \log (nT))$\\

\subsection{Generalization to $Pm \mid L_{\max}$}

To generalize to $m$ machines, we number the machines $1,2, \dots m$ and let $DP[i][(x_1, x_2, \cdots, x_{m-1})]$ indicate whether there is a feasible scheduling of the first $i$ jobs jobs such that machine $j$ completes all jobs scheduled to it in exactly $x_j$ time for $1 \leq j \leq m-1$ (and this uniquely determines the completion time for machine $m$ since it's just the remaining time that isn't accounted for in the first $m-1$ machines).\\

We again assume that $t_i$ is the amount of time it takes to complete job $i$ and $d_i$ is the deadline of job $i$. To compute $DP[i][(x_1, x_2, \cdots, x_{m-1})]$, if $DP[i-1][(x_1, x_2, \cdots, x_j-t_i, \cdots, x_{m-1}] = 1$ for any $j$ such that $x_j \le d_i$, then we have a feasible scheduling of the first $i-1$ jobs such that if we add the $i^{th}$ job to the end of the schedule for the $j^{th}$ machine, we get the desired completion times for all of the machines and it is still feasible.\\

In the other case, we add the $i^{th}$ job to the last machine. Then, we need to check that $DP[i-1][(x_1, x_2, \cdots, x_{m-1})] = 1$, so there is a feasible scheduling of the first $i-1$ jobs with the same weights on the first $m-1$ machines. Then, for this scheduling we also need to make sure that adding the $i^{th}$ job to the last machine does not go over the deadline for that job. The last machine runs in time $\sum_{k=1}^i t_k - \sum_{k=1}^{m-1} x_k$, and we can make this transition as long as $\sum_{k=1}^i t_k - \sum_{k=1}^{m} x_k \le d_i$.\\

Again, we can compute and store $\sum_{k=1}^i t_k$ each time we get to a new row in $O(n^2)$ time. Then, each dimension is at most $nT$ where $T$ is the maximum time it takes to complete a job, so there are $n^mT^{m-1}$ total states, and computing the value of a state takes $O(m)$ time (as we need to check a transition for every single machine and potentially add up $m$ numbers), so the overall runtime is $O(mn^mT^{m-1})$ time, which is pseudopolynomial if $m$ is a constant. \\

\subsection{Generalization to $Qm \mid L_{\max}$}

This time, we will suppose that we are given for machine $i$ a rate parameter $\lambda_i$ which is the number of seconds it takes the machine to do $1$ unit of work. Just like above, we will suppose that job $i$ takes $t_i$ work and has a deadline of $d_i$. We will first scale the rate parameters and the amount of work each job takes until the work each job takes is an integer for every job since scaling both by the same amount doesn't change anything. Then, for the rate parameters and deadlines, if we scale both of these by the same amount then we don't change anything about the problem so we if we assume these numbers are all integers, our algorithm is pseudopolynomial in both the job lengths and machine rates.\\

Again, we let $DP[i][(x_1, x_2, \cdots, x_{m-1})]$ indicate whether there is a feasible arrangement of jobs such that machine $i$ runs in exactly $x_i$ time but here compute the value of $DP[i][(x_1, x_2, \cdots, x_{m-1})]$ as follows:\\

if $DP[i-1][(x_1, x_2, \cdots, x_j-t_i\lambda_j, \cdots, x_{m-1}] = 1$ such that $x_j \le d_i$ for any $1 \le j \le m-1$, we can add the $i^{th}$ job to the end of the schedule of the $j^{th}$ machine and get a feasible scheduling just like above. Otherwise, if we add it to the $m^{th}$ machine, we know that the machine runs in time
$$\lambda_m\left(\sum_{k=1}^i t_k - \sum_{k=1}^{m-1}\frac{x_k}{\lambda_k}\right)$$
as from how we defined rates, if a machine has rate $\lambda_k$ and runs for $x_k$ time, then the total amount of work it does is $\frac{x_k}{\lambda_k}$. The total amount of work to do given the first $i$ jobs is $\sum_{k=1}^i t_k$, so we can subtract these values to get the amount of work that the last machine does and multiply by its rate to find the time it takes to complete all these jobs. Then, we just need to check if this value is less than or equal to $d_i$ to check if this is a valid scheduling since we know that the scheduling of the first $i-1$ jobs is feasible by induction. As there are the same number of states as in the $Pm \mid L_{\max}$ case and each state still takes $O(m)$ time to compute, the runtime is still $O(mn^mT^{m-1})$ and thus pseudopolynomial.

\subsection{Hardness of $P \mid L_{\max}$ and constant-factor approximations}
\begin{theorem}
Bin Packing reduces to $P \mid L_{\max}$
\end{theorem}
\begin{proof}
Given $n$ items, we can scale all of the sizes of the items and the sizes of the bins up until all of the numbers are integers (since arbitrary precision real numbers cannot even be stored, this is always possible). Now, after scaling, suppose the sizes of the bins is $b$ and suppose we have a solution to $P \mid L_{\max}$ in time $T(n)$. This solution must be able to check feasibility of a set of jobs since that's equivalent to checking whether $L_{\max} = 0$.\\

We can now note that if we have $m$ machines and $n$ jobs with lengths equal to the sizes of the items we wished to place in the bins after scaling, all of which have deadline $b$, checking for feasibility of these $n$ jobs is equivalent to checking whether it is possible to pack all of these items into $m$ bins.\\

Now, we can just find the correct number of bins as follows:\\
We can start with a lower bound of $1$ bin and an upper bound of $n$ bins since each item can only be placed in $1$ bin. Then, we can binary search for the correct number of bins. For each value $m$ that we check, we check feasibility of job scheduling on $m$ machines with the jobs described above. If it is not feasible, then this $m$ is a strict lower bound on the number of bins necessary and if it is feasible, then this $m$ is a non-strict upper bound on the number of bins necessary.\\

Each check for feasibility takes $T(n)$ time and we make $\log n$ checks since we can binary search, so this runtime is $T(n) \cdot O(\log n) = Poly(n) \cdot T(n)$. Thus, Bin packing reduces to $P \mid L_{\max}$.
\end{proof}

Since Bin Packing is strongly NP-hard, this means that $P \mid L_{\max}$ is also strongly NP-hard so if we do not assume there are a constant number of machines, we cannot even get a pseudopolynomial algorithm.\\

\begin{theorem}
There is no algorithm which gives us a constant factor approximation to $P2 \mid L_{\max}$ unless $P = NP$. Thus there is also no PTAS to solve $P2 \mid L_{\max}$ unless $P = NP$.
\end{theorem}
\begin{proof}
Firstly, if there were a PTAS to solve $P2 \mid L_{\max}$, for any $\epsilon$ we plug in we have to get a poly-time algorithm that would give us a $(1+\epsilon)$-approximation which is a constant factor approximation, so the fact that there is no PTAS follows directly from the lack of a constant factor approximation.\\

To see that there is no constant factor approximation, note that if we had a $(1 + \epsilon)$ approximation, when the answer is $0$, $(1+\epsilon)0 = 0$ so we need to get an answer of $0$ from our approximation algorithm. This means that in order to get a constant factor approximation in polynomial time, we need to be able to solve the decision problem for whether a set of jobs can be scheduled feasibly in polynomial time. However, as we saw above, solving the problem exactly has a polynomial time reduction to this. Thus, if this were possible, it would also be possible to solve $P2 \mid L_{\max}$ exactly in polynomial time. However, $P2 \mid L_{\max}$ is NP-hard so if we could do this, then we would have $P = NP$.\\

Thus, unless $P = NP$ there is no constant factor approximation or PTAS for\\
$P2 \mid L_{\max}$.
\end{proof}

\section{Conclusion}
Previously, while it was known that all of the above problems were $NP$-hard, and that a slightly harder version of the easiest problem above was strongly $NP$-hard, we did not know whether $P2 \mid \mid L_{\max}$ was is $P_{pseudo}$ or if it was strongly $NP$-hard. In this paper, we have resolved this by showing that it is in $P_{pseudo}$. We further proved that we could generalize this to show that $Qm \mid \mid L_{\max}$ was also in $P_{pseudo}$.\\

We also showed some hardness results: we showed that $P \mid \mid L_{\max}$, which is a different way to make $P2 \mid \mid L_{\max}$ harder than what was looked at before, was also strongly $NP$-hard, and also showed that there cannot exist any constant factor approximation schemes for these problems unless $P = NP$.\\

As such, a pseudopolynomial algorithm to solve $Qm \mid \mid L_{\max}$ is in a sense optimal for deterministic algorithms for this problem unless $P = NP$ as we cannot get any polynomial time constant factor approximation. Furthermore, many natural extensions of this problem end up being strongly $NP$-hard, so this is in a sense "close" to the boundary between weakly $NP$-hard problems and strongly $NP$-hard problems.\\

\newpage

\bibliographystyle{unsrt}
\bibliography{main}

\end{document}